\documentclass[12pt]{article}
\usepackage[utf8]{inputenc}
\usepackage[english]{babel}
\usepackage{amsmath}
\usepackage{amsthm}
\usepackage{amssymb}
\usepackage{hyperref}
\usepackage{braket}
\usepackage{color}
\usepackage{comment}
\usepackage{graphicx}
\usepackage[margin=2.5cm]{geometry}
\usepackage{authblk}
\usepackage{todonotes}
\usepackage[sort&compress,numbers]{natbib}
\usepackage{enumerate}

\usepackage{stackengine}
\usepackage[noabbrev]{cleveref}
\usepackage{mathtools}
\usepackage{mathrsfs}
\usepackage{bbm}

\newcommand{\epsi}{\varepsilon}
\newcommand{\E}{{\mathrm{e}}}
\newcommand{\D}{\mathrm{d}}
\newcommand{\I}{\mathrm{i}}
 \newcommand{\R}{ \mathbb{R} }
  \newcommand{\Sph}{ \mathbb{S} }

\newcommand{\N}{ \mathbb{N} }

\DeclareMathOperator{\spec}{spec}

\newcommand{\norm}[1]{\left\Vert #1 \right\Vert}
\newcommand{\abs}[1]{\left\vert #1 \right\vert}

\newcommand{\longip}[3]{\left\langle #1 \middle\vert #2 \middle\vert #3 \right\rangle}
\newcommand{\ud}{\,\textnormal{d}}
\let\eps\varepsilon

\allowdisplaybreaks
\theoremstyle{plain}
    \newtheorem{thm}{Theorem}
	\newtheorem{lem}[thm]{Lemma}
	\newtheorem{cor}[thm]{Corollary}
	\newtheorem{prop}[thm]{Proposition}
	
\theoremstyle{definition}
    \newtheorem{defi}[thm]{Definition}
    \newtheorem{rmk}[thm]{Remark}

\title{The BCS Energy Gap at High Density}
\author{Joscha Henheik\footnote{\href{mailto:joscha.henheik@ist.ac.at}{joscha.henheik@ist.ac.at}}~~and Asbj{\o}rn B\ae kgaard Lauritsen\footnote{\href{mailto:asbjornbaekgaard.lauritsen@ist.ac.at}{alaurits@ist.ac.at}} \\IST Austria, Am Campus 1, 3400 Klosterneuburg, Austria}

\begin{document}

\maketitle

\begin{abstract} 
We study the BCS energy gap $\Xi$ in the high--density limit and derive an asymptotic formula, 
which strongly depends on the strength of the interaction potential $V$ on the Fermi surface. 
In combination with the recent result by one of us (Math.~Phys.~Anal. Geom. 25, 2022) on the critical temperature $T_c$ at high densities, 
we prove the universality of the ratio of the energy gap and the critical temperature.
\end{abstract}

\section{Introduction and Main Results}  
The Bardeen--Cooper--Schrieffer (BCS) theory \cite{bcs} (see \cite{hs15} for a review of recent rigorous mathematical work) 
has been an important theory of superconductivity since its conception. 
More recently, it has also gained attraction for describing the phenomenon of superfluidity in ultra cold fermionic gases, see \cite{bloch,chen} for reviews.
In either context, BCS theory is often formulated in terms of the BCS gap equation (at zero temperature)
\begin{equation}
\label{gapeq} \Delta(p) = -\frac{1}{(2\pi)^{3/2}} \int_{\R^3} \hat{V}(p-q) \frac{\Delta(q)}{E_{\Delta,\mu}(q)}  \D q\,,
\end{equation} 
where $E_{\Delta,\mu}(p) = \sqrt{(p^2-\mu)^2+ \vert \Delta(p) \vert^2}$.
At finite temperature $T > 0$ one replaces $E_{\Delta, \mu}$ by $E_{\Delta,\mu}/\tanh( E_{\Delta, \mu} / 2T)$.
The function $\Delta$ is interpreted as the order parameter describing the Cooper pairs (paired fermions). 
The interaction is local and given by the potential $V$, which we will assume satisfies $V \in L^1(\R^3)$, in which case it
has a Fourier transform given by $(\mathfrak{F}V)(p) = \hat{V}(p) = (2\pi)^{-3/2} \int_{\R^3} V(x) \E^{- \I p\cdot x} \ud x$.

The chemical potential $\mu$ controls the density of the fermions, and we investigate the high--density limit, i.e.~$\mu \to \infty$, here.
Recently this limit was studied by one of us \cite{henheikTc},
where an asymptotic formula for the critical temperature $T_c$ was found. 
 For temperatures $T$ below the critical temperature, $T < T_c$, the gap equation at temperature $T$ (\Cref{gapeq} with $E_{\Delta, \mu}$ replaced as prescribed) 
admits a non--trivial solution, for $T \ge T_c$ it does not. The critical temperature may equivalently be characterized by the existence of a negative eigenvalue of a certain linear operator, see \cite{hhss}.
Physically, a system at temperature $T$ is superconducting/--fluid if $T < T_c$, if $T \ge T_c$ it is not.

In this paper we study the energy gap (at zero temperature)
\begin{equation}
\Xi = \inf_p E_{\Delta,\mu}(p) = \inf_p \sqrt{(p^2-\mu)^2+ \vert \Delta(p)\vert^2}\,. \label{eq:energygap}
\end{equation}
The function $E_{\Delta,\mu}$ has the interpretation of the dispersion relation for the corresponding BCS Hamiltonian, 
and so $\Xi$ is indeed an energy gap (see Appendix A in \cite{hhss}). 
We show that, in the high--density limit, $\mu \to \infty$,
the ratio of the energy gap and the critical temperature tends to a universal constant independent of the interaction potential,
\begin{equation} \label{eq:universality}
\frac{\Xi}{T_c} \approx \frac{\pi}{\E^{\gamma}}\,,
\end{equation}
where $\gamma \approx 0.577$ denotes the Euler--Mascheroni constant. 
This universality is well--known in the physics literature, see, e.g., \cite{gorkov}, and was rigorously verified in the weak--coupling limit by Hainzl and Seiringer \cite{hs081}
and in the low--density limit, $\mu \to 0$, by one of us \cite{lauritsen} building on a work by Hainzl and Seiringer \cite{hs08}.
The general strategy for proving the universality in these limits has been to establish sufficiently good asymptotic formulas
for both, $T_c$ and $\Xi$, and compare them afterwards.

The weak--coupling limit is studied in \cite{hs081,fhns}, where one considers a potential $\lambda V$ 
for $V$ fixed and a small coupling constant $\lambda\to 0$.
In this limit, Hainzl and Seiringer \cite{hs081} have shown that the critical temperature and energy gap satisfies $T_c \sim A \exp(-B/\lambda)$ 
and $\Xi \sim C \exp(-B/\lambda)$ respectively for explicit constants $A,B,C > 0$ depending on the interaction potential $V$ and the chemical potential $\mu$. 
This limit exhibits the same universality and the ratio $C/A = \pi \E^{-\gamma}$ is independent of the interaction potential $V$
and the chemical potential~$\mu$.

The low--density limit $\mu \to 0$ is studied in \cite{hs08,lauritsen}.
In this limit Hainzl and Seiringer \cite{hs08} have shown that the critical temperature satisfies $T_c \sim \mu A \exp(-B/\sqrt{\mu})$
and one of us \cite{lauritsen} has shown that the energy gap satisfies $\Xi \sim \mu  C \exp(-B/\sqrt{\mu})$,
for some (different) explicit constants $A,B,C > 0$ depending on the interaction potential $V$. 
Also in this limit we have the same universality and the ratio $C/A = \pi \E^{-\gamma}$ is independent of the interaction potential $V$.  
These results together with the present paper thus show that the universality~\eqref{eq:universality}
holds in both, the low-- and high--density limit, as well as in the weak--coupling limit.

To show the universality, we prove in \Cref{thm:1} an asymptotic formula for the energy gap $\Xi$
in the high--density limit, similar to the corresponding formula for the critical temperature given in Theorem 7 in \cite{henheikTc}.
This formula, as well as the one given in Theorem \ref{thm:1}, depends strongly on the strength of the interaction potential $V$ on the Fermi sphere $\{p^2 = \mu\}$, which becomes weak due to the decay of $\hat V$ in momentum space. 
Together with the formula for the critical temperature \cite{henheikTc}
we prove the universality \eqref{eq:universality} in \Cref{cor:1}. All proofs are given in \Cref{sec:proofs}.
We now introduce some technical constructions and give the precise statements of our results.

\subsection{Preliminaries}
We will work with the formulation of BCS theory of 
\cite{hs08,hs081,fhns,hhss,hs15,henheikTc,lauritsen}. There 
one considers minimizers of the BCS functional (at zero temperature)
\begin{equation}\label{eqn.bcs.functional}
\mathcal{F}(\alpha)
 = 
    \frac{1}{2} \int_{\R^3} |p^2 - \mu| \left(1 - \sqrt{1 - 4|\hat \alpha(p)|^2}\right) \ud p + \int_{\R^3} V(x) |\alpha(x)|^2 \ud x\,.
\end{equation}
If $\alpha$ is a minimizer of this, then $\Delta = -2\widehat{V\alpha}$ 
satisfies the BCS gap equation \eqref{gapeq}. 
As discussed in \cite{hs081} the minimizer $\alpha$ is in general not necessarily unique, hence also $\Delta$ and $\Xi$ are not necessarily unique. 
However, since we  will assume that the interaction $V$ has non--positive Fourier transform, $\alpha$ and thus $\Xi$ is unique (see Lemma~2 in \cite{hs081}).

A crucial role for the investigation of the energy gap \eqref{eq:energygap} in the high--density limit 
is played by the (rescaled) operator  $\mathcal{V}_{\mu} : L^2(\mathbb{S}^{2}) \to L^2(\mathbb{S}^{2})$ 
measuring the strength of the interaction potential $\hat{V}$ on the Fermi surface. 
It is defined as
\begin{equation} \label{eq:vmu}
\left(\mathcal{V}_{\mu} u\right)(p) = \frac{1}{(2\pi)^{3/2}} \int_{\Sph^{2}} \hat{V}(\sqrt{\mu} (p-q)) u(q)\, \D\omega(q) \,,
\end{equation}
where $\D \omega$ denotes the uniform (Lebesgue) measure on the unit sphere $\Sph^2$.
The pointwise evaluation of $\hat{V}$ (and in particular on a $\mathrm{codim-}1$ submanifold) 
is well defined since $V \in L^1(\R^3)$. 
The condition that $V \in L^1(\R^3)$ could potentially be relaxed, see \cite{cueninmerz} and Remark~9 in \cite{henheikTc}.
The lowest eigenvalue of $\mathcal{V}_{\mu}$, which we denote by
\begin{equation*}
e_{\mu} = \mathrm{inf\, spec}\, \mathcal{V}_{\mu}
\end{equation*}
will be of particular importance. 
Note, that $\mathcal{V}_{\mu}$ is a trace--class operator (see the argument above Equation (3.2) in \cite{fhns}) with 
\begin{equation*}
\mathrm{tr}(\mathcal{V}_{\mu}) = \frac{1}{2\pi^2} \int_{\R^3} V(x) \D x = \sqrt{\frac{2}{\pi}} \, \hat{V}(0)\,.
\end{equation*} 
We will assume that $\hat V(0) < 0 $ in which case $e_{\mu} < 0$.
This corresponds to an attractive interaction between (some) electrons on the Fermi sphere.

In this work, we  restrict ourselves to the special case of radial potentials $V$,
where the spectrum of $\mathcal{V}_{\mu}$ can be determined more explicitly (see, e.g., Section~2.1 in~\cite{fhns}). 
Indeed, for radial $V$, 
the eigenfunctions of $\mathcal{V}_{\mu}$ are spherical harmonics and the corresponding eigenvalues are 
\begin{equation} \label{eq:eigenvalues}
\frac{1}{2\pi^2} \int_{\R^3} V(x) \left(j_\ell(\sqrt{\mu}\vert x \vert)\right)^2 \D x\,.
\end{equation}
The lowest eigenvalue $e_\mu$ is thus given by
\begin{equation*}
e_\mu = \frac{1}{2\pi^2} \, \inf_{\ell \in \mathbb{N}_0} \, \int_{\R^3} V(x) \left(j_\ell(\sqrt{\mu}\vert x \vert)\right)^2 \D x \,.
\end{equation*} 
Here, $j_\ell$ denotes the spherical Bessel function of order $\ell \in \mathbb{N}_0$.
Additionally, in case that $\hat{V} \le 0$, we have, by the Perron--Frobenius theorem, 
that the minimal eigenvalue is attained for the constant eigenfunction (i.e. with $\ell = 0$). Thus
\begin{equation} \label{eq:emuangmom0}
e_{\mu} = \frac{1}{2\pi^2} \int_{\R^3} V(x) \left(\frac{\sin(\sqrt{\mu}\vert x \vert)}{\sqrt{\mu}\vert x \vert}\right)^2 \D x\,.
\end{equation} 
For further discussions of the radiality assumption on $V$, see Remark~8 in \cite{henheikTc}.

In order to obtain an asymptotic formula for the energy gap that is valid up to second order (see \cite{hs081, henheikTc}), 
we define the operator $\mathcal{W}_\mu^{(\kappa)}$ on $u \in L^2(\Sph^2)$ via its quadratic form 
\begin{align} \label{eq:defWmu}
\big\langle u \big\vert\mathcal{W}_\mu^{(\kappa)} \big\vert u \big\rangle = \sqrt{\mu}\int_{0}^{\infty} \D \vert p \vert &\left(\frac{\vert p \vert^2}{\vert \vert p \vert^2 - 1 \vert} \left[\int_{\Sph^{2}}\D \omega(p) \left(\vert \hat{\varphi}(\sqrt{\mu} p) \vert^2 - \vert \hat{\varphi}(\sqrt{\mu} p/\vert p \vert) \vert^2\right)\right]\right. \nonumber \\
&+ \left. \frac{\vert p\vert^2}{\vert p\vert^2 + \kappa^2} \int_{\Sph^{2}} \D \omega(p) \vert \hat{\varphi}(\sqrt{\mu} p/\vert p \vert) \vert^2
\right)
\end{align}
for any fixed $\kappa \ge 0$ (cf.~Equation (10) in \cite{henheikTc} resp.~Equation (13) in \cite{hs081} for an analogous definition with $\kappa = 0$). 
Here $\hat{\varphi}(p) = (2\pi)^{-3/2} \int_{\Sph^{2}} \hat{V}(p-\sqrt{\mu}q) u(q) \D \omega(q)$, 
and $(\vert p\vert , \omega(p)) \in (0,\infty) \times \Sph^2$ denote spherical coordinates for $p \in \mathbb{R}^3$. 
To see that this operator is well--defined note that 
the map $\vert p \vert \mapsto \int_{\Sph^{2}} \D \omega(p) \vert \hat{\varphi}(p)\vert^2$ 
is Lipschitz continuous for any $u \in L^2(\Sph^2)$ since $V \in L^1(\R^3)$.
Hence the radial integral in \Cref{eq:defWmu} is well defined for $\vert p\vert \sim 1$. 
We will further assume that $V \in L^{3/2}(\R^3)$, in which case the integral is well--defined
for large $\vert p \vert$ as well.
We formulate our result in Theorem \ref{thm:1} only for $\kappa = 0$, but the case of a positive parameter $\kappa >0$ is crucial in the proof of this statement. For example, $\kappa >0$ ensures that the second term in the decomposition of the Birman--Schwinger operator associated with $E_{\Delta,\mu} + V$ is small (cf.~\Cref{bsdecomp}). Whenever it does not lead to confusion, we will refer to some $\kappa$--dependent quantity at $\kappa = 0$ by simply dropping the  $(\kappa)$--superscript.

We now define the operator
\begin{equation} \label{eq:Bmu}
\mathcal{B}_\mu^{(\kappa)} = \frac{\pi}{2}\left(\mathcal{V}_\mu - \mathcal{W}_\mu^{(\kappa)}\right)\,,
\end{equation}
which captures the strength of the interaction potential near the Fermi surface to second order and denote its lowest eigenvalue by 
\begin{equation} \label{eq:bmu}
b_\mu^{(\kappa)} = \mathrm{inf \, spec} \, \mathcal{B}_\mu^{(\kappa)}\,.
\end{equation}
The factor $\pi/2$ is introduced in \Cref{eq:Bmu} since for this scaling, 
the eigenvalue $b_\mu^{(\kappa)}$ has the interpretation of an effective scattering length in the case of small $\mu$ 
(see Proposition~1 in \cite{hs081}). 
Moreover, it was shown during the proof of Theorem 7 in \cite{henheikTc} that if $e_\mu <0$ then also $b_\mu^{(\kappa)}<0$ for $\mu$ large enough. 
This will also follow from \Cref{eq:Wmubound} in the proof below.

\subsection{Results}
The following definition characterizes the class of interaction potentials for which our asymptotic formula will hold. 
\begin{defi}[Admissible potentials]
\label{def:admissible}
Let $V \in L^1(\R^3) \cap L^{3/2}(\R^3)$ be a radial real--valued function with non--positive Fourier transform $\hat V \le 0$ and $\hat V(0) <0$.
 Denote
\begin{equation}
\label{eq:defsstar}
  s_\pm^* := \sup \left\{ s \geq 0 : |\cdot|^{-s} V_\pm \in L^1(\R^3)\right\},
  \qquad
  s^* := \min \{s_+^*, s_-^*\}\, ,
\end{equation}
where $V_\pm = \max\{ \pm V, 0\}$ are the positive and negative parts of $V$.
We say that $V$ is \emph{admissible} if the following is satisfied:
\begin{enumerate}[(a)]
\item 
There exists $a > 0$ such that
\[
  \sup \left\{ r \geq 0 : \lim_{\eps \to 0} \frac{1}{\eps^r}\int_{B_\eps} V_\pm(x) \ud x = 0\right\}
   = \sup \left\{ r \geq 0 : \lim_{\eps \to 0} \frac{1}{\eps^r} \int_{B_\eps} V_\pm|_{B_{a}}^*(x) \ud x = 0 \right\}\, ,
\]
where $V_\pm|_{B_{a}}^*$ 
denotes the symmetric decreasing rearrangement of $V_\pm|_{B_{a}}$, the restriction of $V_\pm$ to the ball of radius $a$ around $0$,
\item
if $|\cdot|^{-2}V \notin L^1(\R^3)$, we have $s^* = s_-^* < s^*_+$, and
\item 
$\vert \cdot \vert V \in L^2(\R^3)$ 
and $s^* > 7/5$.
\end{enumerate}
\end{defi}

\noindent
	As discussed around \Cref{eqn.bcs.functional}, the definiteness of the Fourier transform is needed for ensuring uniqueness of the energy gap $\Xi$.
Intuitively, the other criteria may be though as follows: Assumption~(a) captures that the strongest singularity 
of $V$ near the origin is in fact at the origin, assumption~(b) captures that $V$ is predominantly attractive, and assumption~(c) captures that $V$ is slightly less divergent at the origin, than allowed by the $L^{3/2}(\R^3)$-assumption. 
In view of assumption~(a), we remark that it is natural that the system is sensitive to the short range behavior of the interaction potential, since the interparticle distance as the physically relevant length scale that depends on the particle density tends to zero in the high--density limit. Furthermore, note that for $V \in L^1(\R^3) \cap L^{3/2}(\R^3)$, the condition $\vert \cdot \vert V \in L^2(\R^3)$ is mainly about regularity away from $0$ and infinity.

The most important examples of allowed interaction potentials include the cases of attractive Gaussian, Lorentzian and Yukawa potentials, also discussed in \cite{langmann}.
That is
\[
  V_{\textnormal{Gauss}}(x) = - (2\pi)^{-3/2}\E^{-x^2/2}\,, \quad V_{\textnormal{Lorentz}}(x) = - \frac{1}{\pi^2(1 + x^2)^2}\,,
  \quad V_{\textnormal{Yukawa}}(x) = - \frac{1}{4\pi|x|}\E^{-|x|}\,.
\]
\begin{rmk} \label{rmk:f(s)}
The proof of our main result formulated in Theorem \ref{thm:1} works without change if we assume $|\cdot|V\in L^r(\R^3)$ for some $2 \ge r > f(s^*)$
instead of $|\cdot|V\in L^2(\R^3)$, where $f$ is some complicated (explicit) expression, see the proof of \Cref{prop:4}.
We do not state the theorem with this slight generalization for simplicity.
We will however give the proof under this more general assumption for the purpose of illuminating where the assumption on $r=2$ comes from.
Additionally, to further illuminate where the conditions are used, all propositions and lemmas are stated with only the 
conditions needed on $V$ for that specific statement. 
(Beyond the conditions that $V\in L^1(\R^3)\cap L^{3/2}(\R^3)$ is real--valued, radial and has $\hat V \leq 0, \hat V(0) < 0$,
which is always assumed.)
\end{rmk}

\noindent
We can now state our main result for admissible interaction potentials. 
\begin{thm}\label{thm:1} Let $V$ be an admissible potential. Then the energy gap $\Xi$ is positive and satisfies
  \begin{equation} \label{eq:thm2}
  \lim\limits_{\mu \to \infty} \left(\log\frac{\mu}{\Xi} + \frac{\pi}{2\sqrt{\mu} b_\mu}\right) = 2- \log(8)\,.
  \end{equation} 
\end{thm}

\noindent
In other words, 
\begin{equation*}
  \Xi = \mu \left(8\, \E^{-2} + o(1)\right)\exp\left( \frac{\pi}{2 \sqrt{\mu} b_\mu}\right)
\end{equation*}
in the limit $\mu \to \infty$. Similarly as for the critical temperature \cite{henheikTc}, this asymptotic formula is completely analogous to the weak--coupling case \cite{hs081} (replace $V \to \lambda V$ and take the limit $\lambda \to 0$) but we have coupling parameter $\lambda =1$ here. This similarity is not entirely surprising. From a physical perspective, only those fermions with momenta close to the Fermi surface $\{p^2 = \mu \}$ contribute to the superconductivity/--fluidity.
Thus, by the decay of the interaction $\hat{V}$ in Fourier space, the high--density limit, $\mu \to \infty$, is effectively a weak--coupling limit.

In order to deduce universality as in \Cref{eq:universality} in the high--density limit, 
we show that every admissible potential in the sense of Definition \ref{def:admissible} satisfies the imposed conditions for the proof of an analogous formula for the critical temperature. 
These conditions were formulated in Definition~5 in \cite{henheikTc}.
\begin{prop} \label{prop:admissible} 
Every admissible potential satisfies the conditions of Definition~5 in \cite{henheikTc}.
\end{prop}
\begin{proof}
By comparing the two definitions, the statement is trivial apart from the following two points. 
First, the additional requirement $\int_{\R^3} \frac{V(x)}{\vert x \vert^2} \D x < 0$ from Definition~5 in \cite{henheikTc} in the case $|\cdot|^{-2}V \in L^1(\R^3)$ is automatically fulfilled, 
since
\begin{equation*}
-\Delta_p \widehat{\frac{V}{\vert \cdot \vert^2}}(p) = \hat{V}(p) \le 0\,.
\end{equation*} 
That is, the radial function $\widehat{\frac{V}{\vert \cdot \vert^2}}$ is subharmonic and approaches $0$ as $\vert p \vert \to \infty$ (by the Riemann--Lebesgue Lemma), 
and thus by the maximum principle assumes a strictly negative value at~$0$. 
Second, since $\hat V \le 0$ and by application of the Perron--Frobenius Theorem, 
the constant spherical harmonic is the unique normalized ground state of $\mathcal{V}_\mu$ and thus condition (d) from 
Definition~5 in \cite{henheikTc} can be dropped.
\end{proof}

\noindent Therefore, by means of  Theorem~7 in \cite{henheikTc}, the critical temperature $T_c$ satisfies
\begin{equation*}
T_c = \mu \left(\frac{8}{\pi}\E^{\gamma-2} + o(1)\right)\exp\left(\frac{\pi}{2 \sqrt{\mu} b_\mu}\right)
\end{equation*}
for any admissible potential. Here $\gamma \approx 0.577$ is the Euler--Mascheroni constant. Together with Theorem \ref{thm:1}, this immediately proves the following.
\begin{cor}\label{cor:1} Let $V$ be an admissible potential. Then 
\begin{equation*}
\lim\limits_{\mu \to \infty} \frac{\Xi}{T_c} = \frac{\pi}{\E^\gamma} \approx 1.764\,.
\end{equation*}
\end{cor}

\noindent
This universality of the ratio between the energy gap and the critical temperature is well known in the physics literature (see, e.g., \cite{gorkov}) and has been previously established rigorously in the weak--coupling and low--density limits (see \cite{hs081} resp.~\cite{lauritsen}).

\section{Proofs} \label{sec:proofs}
As in the analysis of the critical temperature \cite{henheikTc} we introduce the parameter $\kappa >0$.
We have the following comparison of $b_\mu^{(\kappa)}$ with the $\kappa = 0$ quantity.
\begin{lem}[{\cite[Lemma~15]{henheikTc}}] \label{lem:2}
Let $V$ be admissible and $\kappa > 0$. 
In the limit of high density, $\mu \to \infty$, we have
\begin{equation*}
  \frac{\pi}{2\sqrt{\mu} b_\mu} = \frac{\pi}{2\sqrt{\mu} b_\mu^{(\kappa)}} +  \kappa \, \frac{\pi}{2} + o(1)\,.
\end{equation*}
\end{lem}
\begin{proof} 
This is immediate from Lemma 15 in \cite{henheikTc} by invoking Proposition \ref{prop:admissible}. 
\end{proof}
\noindent
Now, one important ingredient in our proof is the asymptotic behavior of 
\begin{equation*}
m_{\mu}^{(\kappa)}(\Delta) = \frac{1}{4\pi} \int_{\R^3} \left(\frac{1}{E_{\Delta,\mu}(p)} - \frac{1}{p^2+\kappa^2\mu}\right)\D p 
\end{equation*}
for fixed $\kappa > 0$ (recall that $E_{\Delta,\mu}(p) = \sqrt{(p^2-\mu)^2+\vert \Delta(p)\vert^2}$).
This is similar to the strategy for the weak--coupling, low--density, and high--density limits of the critical temperature (see \cite{hs081, hs08, henheikTc}), 
and for the weak--coupling and low--density limits of the energy gap (see \cite{hs081, lauritsen}). 

\begin{lem} \label{lem:1}
Let $V$ be admissible and $\kappa > 0$. In the limit of high density, $\mu \to \infty$, we have
\begin{align*}
\Xi &= \Delta(\sqrt{\mu})(1+o(1))\,, \\[1.5mm]
m_\mu^{(\kappa)}(\Delta) &= \sqrt{\mu} \left( \log\frac{\mu}{\Delta(\sqrt{\mu})}-2+\kappa \, \frac{\pi}{2} + \log(8) + o(1) \right)\,,\\[1.5mm]
\frac{m_\mu^{(\kappa)}(\Delta)}{\sqrt{\mu}} &= - \frac{\pi}{2\sqrt{\mu}b_\mu^{(\kappa)}} + o(1)\,.
\end{align*}
\end{lem}

\noindent
These three asymptotic equalities are proven in Propositions \ref{prop.xi.equal.delta}, \ref{prop:4}, and \ref{prop:5} respectively.
\begin{proof}[Proof of Theorem \ref{thm:1}]
By Lemma \ref{lem:1} and Lemma \ref{lem:2} we get
\begin{align*}
\lim\limits_{\mu \to \infty} &\left(\log \frac{\mu}{\Xi} +  \frac{\pi}{2\sqrt{\mu}b_\mu}\right) = \lim\limits_{\mu \to \infty} \left(\log \frac{\mu}{\Delta(\sqrt{\mu})} + \frac{\pi}{2\sqrt{\mu}b_\mu}\right) \\[2mm] &= \lim\limits_{\mu \to \infty} \left(\log \frac{\mu}{\Delta(\sqrt{\mu})} + \frac{\pi}{2\sqrt{\mu}b_\mu^{(\kappa)}}\right) + \kappa\, \frac{\pi}{2}  = \lim\limits_{\mu \to \infty} \left(\log \frac{\mu}{\Delta(\sqrt{\mu})} - \frac{m_\mu^{(\kappa)}(\Delta)}{\sqrt{\mu}}\right) + \kappa \, \frac{\pi}{2} \\[2.5mm] &= \ 2-\kappa \, \frac{\pi}{2} + \log(8)+ \kappa \, \frac{\pi}{2} = 2-\log(8)\,,
\end{align*}
which yields \eqref{eq:thm2} and we have proven Theorem \ref{thm:1}. 
\end{proof}
\noindent
The rest of this paper is devoted to the proof of Lemma \ref{lem:1}.

\subsection{Proof of Lemma \ref{lem:1}}
As remarked, a key idea is to study the integral $m_\mu^{(\kappa)}(\Delta)$.
As in \cite{hs081,lauritsen} we first need some control of $\Delta$
in the form of a Lipschitz--like bound (given in \Cref{prop.delta.lipschitz})
and a bound controlling $\Delta(p)$ in terms of $\Delta(\sqrt{\mu}q)$ for $q\in \Sph^2$ 
(given in \Cref{eqn.delta.leq.delta}). 
First, we recall some properties (from \cite{hs081}) of the minimizer $\alpha$ 
of the BCS functional at zero temperature
\begin{equation} \label{eq:BCSfunctional2}
\mathcal{F}(\alpha)
 = 
    \frac{1}{2} \int_{\R^3}|p^2 - \mu| \left(1 - \sqrt{1 - 4|\hat \alpha(p)|^2}\right) \ud p + \int_{\R^3} V(x) |\alpha(x)|^2 \ud x\,.
\end{equation}
In \cite[Lemma 2]{hs081} it is shown that for potentials $V$ with non-positive Fourier transform
there exists a unique minimizer $\alpha$ with (strictly) positive Fourier transform. 
Moreover, for radial $V$ the BCS functional is invariant under rotations. Hence $\alpha$ and thus also $\Delta = -2\widehat{V\alpha}$ are radial functions.
Therefore, with a slight abuse of notation, we will write $\Delta(|p|)$ and mean $\Delta(p)$ for some (any) vector $p$.
(In general for any radial function $f$, we will write $f(|p|)$ for the value of $f(p)$.)
Additionally, since $\hat V\leq 0$ we have that $\Delta \geq 0$. 
In fact, by the BCS gap equation \eqref{gapeq}, we even have $\Delta > 0$, see Lemma 2 in \cite{hs081}.
Now, we give some \emph{a priori} bounds on the minimizer $\alpha$. The proofs of Lemma \ref{lem:minimizer1} and Lemma \ref{prop.delta.lipschitz} are given in Section \ref{sec:technical}. 

\begin{lem} \label{lem:minimizer1}
Let $\alpha$ be the minimizer of the BCS functional~\eqref{eq:BCSfunctional2}. Then for large $\mu$
\[
  \norm{\alpha}_{L^2}\leq C \mu^{7/20} \quad \text{and} \quad \norm{\alpha}_{H^1} \leq C \mu^{3/4}.
\]
\end{lem}
\noindent These estimates on the minimizer $\alpha$ now translate to bounds on $\Delta = -2\widehat{V\alpha} $.
\begin{lem}
\label{prop.delta.lipschitz}
Suppose $V\in L^{r}(\R^3)$ for some $6/5 \leq r \leq 2$. Define $\delta_r = \frac{3}{4} - \frac{6}{5r}$.
Then for sufficiently large $\mu$ we have
\[
  \norm{\Delta}_{L^\infty} \leq C \mu^{\frac{24-5r}{20r}} = C\mu^{\frac{1}{2} - \delta_r}.
\]
Similarly, if $|\cdot| V \in L^{r}(\R^3)$ then
\[
  |\Delta(p) - \Delta(q)| \leq C \mu^{\frac{24-5r}{20r}} ||p| -|q|| = C\mu^{\frac{1}{2} - \delta_r}||p|-|q||
\]
for all $p,q$.
In particular, if $r > 8/5$ then $\delta_r > 0$ and thus $1/2-\delta_r< 1/2$.
\end{lem}
\noindent
We will use the first bound as $\Vert \Delta \Vert_{L^\infty} \le C \mu^{11/20} = o(\mu)$ for $r = 3/2$, and the second bound as $\vert \Delta(p) - \Delta(q) \vert \le C \mu^{7/20} \vert \vert p\vert - \vert q \vert \vert$ for $r=2$.

Armed with these {\it a priori} bounds on $\Delta$, we can now prove the asymptotic formulas in Lemma~\ref{lem:1} and start with the first one. 

\begin{prop}\label{prop.xi.equal.delta}
Suppose $|\cdot| V \in L^r({\R^3})$ for $r > 8/5$. Then $\Xi = \Delta(\sqrt{\mu}) (1 + o(1))$.
\end{prop}
\begin{proof}
Clearly $\Xi = \inf \sqrt{|p^2 - \mu|^2 + |\Delta(p)|^2}\leq \Delta(\sqrt{\mu})$.
Take now $p$ with $|p^2- \mu| \leq \Xi \leq \Delta(\sqrt{\mu})$. 
Then
\[
  |\Delta(p) - \Delta(\sqrt{\mu})| \leq C \mu^{1/2 -\delta_r} ||p|-\sqrt{\mu}|
    \leq C \mu^{1/2 - \delta_r} \frac{\Delta(\sqrt{\mu})}{|p| + \sqrt{\mu}}
    \leq C \mu^{-\delta_r} \Delta(\sqrt{\mu})
\]
where $\delta_r >0$ by assumption. Hence, $\Delta(p) = \Delta(\sqrt{\mu})(1 + o(1))$ for any such $p$ and we conclude the desired.
\end{proof}
\noindent
 The proofs of the second and third equality (Proposition \ref{prop:4} and Proposition \ref{prop:5}, respectively) heavily use Lemma \ref{lem:3} and Lemma \ref{lem:4}, which we import from \cite{henheikTc}. 
Lemma \ref{lem:3} provides an upper bound for integrals of the potential against spherical Bessel functions $j_\ell$, uniformly in $\ell \in \mathbb{N}_0$. These naturally arise by the spherical symmetry of $V$ (cf.~\Cref{eq:eigenvalues}). 

\begin{lem}{\rm (\cite[Lemma 12]{henheikTc})}\label{lem:3}
Let $V \in L^1(\R^3) \cap L^{3/2}(\R^3)$ and
assume that $s^* >1$, with $s^*$ as in Definition \ref{def:admissible}. Set
\[
  \beta^* 
    = \begin{cases}
      \frac{s^*}{2}  
      & \text{for} \quad s^* \in (1,5/3] 
      \\
      \min\left( \frac{4s^*-4}{9s^*-7}+\frac{1}{2}, \, \frac{19}{22}\right)  \quad 
      & \text{for} \quad s^* >5/3 \,.
  \end{cases}
\]
Note that $\beta^*$ depends continuously on $s^*$ and is (strictly) monotonically increasing (between $1$ and $2$), 
and $\beta^* \le \min(s^*,2)/2$ for any $s^*>1$. Then for any $\delta >0$ there exists an $\epsi_0 >0$ such that for all $\epsi \in [0,\epsi_0]$ we have
\begin{equation*}
\limsup\limits_{\mu \to \infty} \mu^{\beta^* - \delta} \sup_{\ell \in \mathbb{N}_0}\int_{\R^3} \D x \vert V (x)\vert \vert j_\ell(\sqrt{\mu}\vert x \vert )\vert^{2-\varepsilon} = 0\,.
\end{equation*}
\end{lem}   

\noindent
Lemma \ref{lem:4} gives a lower bound on the quantity $e_\mu$ that measures the strength of the interaction potential on the Fermi surface (see \Cref{eq:emuangmom0}). 
\begin{lem}[{\cite[Lemma 13]{henheikTc}}]
\label{lem:4}
	Let $V$ be an admissible potential (cf.~Definition \ref{def:admissible}).  Then for any $\delta >0$ there exists $c_\delta>0$ such that 
	\begin{equation*}
	\liminf\limits_{\mu \to \infty} \vert \mu^{\min(s^*+ \delta ,2)/2} \, e_\mu \vert \ge c_\delta \,.
	\end{equation*} 
\end{lem}
\begin{proof} 
	This is immediate from Lemma 13 in \cite{henheikTc} by invoking Proposition \ref{prop:admissible}. 
\end{proof}
\noindent 
An upper bound is trivially obtained as $\vert e_\mu \vert \le C_\delta \mu^{-\min(s^*-\delta,2)/2}$ for any $\delta >0$ 
by definition of $s^*$ in \Cref{eq:defsstar} (see also \Cref{eq:sVbound}). 
Note that both, upper and lower bound, remain true if we replace the exponent with $\min(s^*,2)/2 \pm \delta$,
i.e.
$c_\delta \mu^{-\min(s^*,2)/2 - \delta}\leq |e_\mu| \leq C_\delta \mu^{-\min(s^*,2)/2 + \delta}$.
This is the formulation we will use.

Beside these two Lemmas, we will use the following observation: It can easily be checked (see Lemma 3 in \cite{hs081}) 
that the operator $E_{\Delta,\mu}(p) + V(x)$ has $0$ as its lowest eigenvalue, and that $\alpha$ is the (unique) eigenvector with this eigenvalue. 
By employing the Birman--Schwinger principle (see \cite{fhns, hhss, hs15}), 
this is equivalent to the fact that the Birman--Schwinger operator
\begin{equation*}
B_{\Delta,\mu} = V^{1/2} \frac{1}{E_{\Delta,\mu}} \vert V \vert^{1/2}
\end{equation*}
has $-1$ as its lowest eigenvalue with $V^{1/2}\alpha$ being the corresponding (unique) eigenvector. 
Here we used the notation $V(x)^{1/2} = \mathrm{sgn}(V(x)) \vert V(x)\vert^{1/2}$. 
In the following we need a convenient decomposition of $B_{\Delta, \mu}$ in a dominant singular term and other error terms. 
For this purpose we let ${\mathfrak{F}_{\mu}} : L^1(\R^3) \to L^2(\Sph^{2})$ denote the (rescaled) Fourier transform restricted to $\Sph^{2}$ with 
\begin{equation*}
\left(\mathfrak{F}_{\mu} \psi \right)(p) = \frac{1}{(2\pi)^{3/2}} \int_{\R^3} \E^{-\mathrm{i} \sqrt{\mu} p\cdot x} \psi(x) \D x\,,
\end{equation*}
which is well--defined by the Riemann--Lebesgue Lemma. Now, we decompose the Birman--Schwinger operator as
\begin{equation} \label{bsdecomp}
B_{\Delta,\mu} = m_\mu^{(\kappa)}(\Delta) \, V^{1/2} {\mathfrak{F}_{\mu}}^\dagger{\mathfrak{F}_{\mu}} \vert V \vert^{1/2} + V^{1/2}  M^{(\kappa)}_{\Delta,\mu}\vert V \vert^{1/2}  \,,
\end{equation}
where $M_{\Delta,\mu}^{(\kappa)}$ is such that this holds. For the first term, note that $V^{1/2} {\mathfrak{F}_{\mu}} ^\dagger{\mathfrak{F}_{\mu}} \vert V \vert^{1/2}$ is isospectral to $ \mathcal{V}_\mu = {\mathfrak{F}_{\mu}} V {\mathfrak{F}_{\mu}}^\dagger$.   In fact, the spectra agree at first except possibly at $0$, but $0$ is in both spectra as the operators are compact on an infinite dimensional space. This first term in the decomposition~\eqref{bsdecomp} will be the dominant term, which is how the  third equality in Lemma \ref{lem:1} will arise. 

Analogously to the proof of Lemma 14 in \cite{henheikTc} and the proof of Theorem 1 in \cite{hs081}, we further decompose
\begin{equation} \label{bsdecomp2}
V^{1/2}  M^{(\kappa)}_{\Delta,\mu}\vert V \vert^{1/2} 
= V^{1/2}  \frac{1}{p^2 + \kappa^2 \mu }\vert V \vert^{1/2} +  A^{(\kappa)}_{\Delta,\mu} 
=: L^{(\kappa)}_{\mu} + A^{(\kappa)}_{\Delta,\mu}\,,
\end{equation}
where now $A^{(\kappa)}_{\Delta,\mu}$ is such that this holds. 
During the proof of Lemma 14 in \cite{henheikTc} (see the  Equation in the middle of page 15) it was shown that
\begin{equation*}
\left\Vert L^{(\kappa)}_{\mu} \right\Vert_{\mathrm{op}} 
\le C \, \mu^{1/2} \,  \int_{0}^{\infty} \D p \frac{p^2}{p^2 + \kappa^2} \sup_{\ell \in \N_0} \int_{\R^3} \D x \vert V(x)\vert \, \vert j_\ell(\sqrt{\mu} p \vert x \vert ) \vert^2 \, ,
\end{equation*}
which may be bounded by $\mu^{-\beta^* + 1/2 + \delta}$ for any $\delta > 0$ by means of \Cref{lem:3}.
We continue with a bound on the operator norm of $A_{\Delta, \mu}^{(\kappa)}$ by estimating the matrix elements $\langle f \vert A_{\Delta, \mu}^{(\kappa)} \vert g \rangle$ 
for functions $f,g\in L^2(\R^3)$. This computation is analogous to the computation in the proof of Theorem 2 in \cite{henheikTc}.
We give it here for completeness.

Note that, since $V$ is radial, it is enough to restrict to functions of definite angular momentum.
That is, with a slight abuse of notation, functions of the form $f(x) = Y_\ell^m(\hat x)f(|x|)$, 
where $Y_\ell^m$ denotes the spherical harmonics and we write $\hat x = x/|x|$.
The operator $A_{\Delta, \mu}^{(\kappa)}$ is indeed block--diagonal in the angular momentum as will follow from the computations below.
Since functions of definite angular momentum span $L^2(\R^3)$ \cite[Sections 17.6-17.7]{hall} it is thus enough
to bound $\langle f \vert A_{\Delta, \mu}^{(\kappa)} \vert g \rangle$
for $f,g$ of the form $f(x) = Y_\ell^m(\hat x)f(|x|)$, $g(x) = Y_{\ell'}^{m'}(\hat x)g(|x|)$.

Now, $A_{\Delta, \mu}^{(\kappa)}$ has integral kernel
\[
  A_{\Delta, \mu}^{(\kappa)}(x,y) 
    = CV^{1/2}(x)|V(y)|^{1/2} \int_{\R^3} \left(\frac{1}{E_{\Delta, \mu}(p)} - \frac{1}{p^2 + \kappa^2\mu}\right)
      \left(\E^{\I p\cdot(x-y)} - \E^{\I \sqrt{\mu}\hat p\cdot (x-y)}\right)  \ud p\,.
\] 
Thus, by the radiality of $V$ we get
\begin{equation}\label{eq.A.innerprod.full}
\begin{aligned}
\big\langle f \big\vert A_{\Delta, \mu}^{(\kappa)} \big\vert g \big\rangle
  & = C\int_0^\infty \ud |x| \, |x|^2 V^{1/2}(|x|) \overline{f(|x|)}
    \int_0^\infty \ud |y| \, |y|^2 |V(|y|)|^{1/2} g(|y|)
  \\ & \quad
    \times 
    \int_{0}^\infty \ud |p| \, |p|^2\left(\frac{1}{E_{\Delta, \mu}(|p|)} - \frac{1}{|p|^2 + \kappa^2\mu}\right)
    \int_{\Sph^2} \ud \omega(\hat p)
  \\ & \quad
    \times
    \int_{\Sph^2} \ud \omega(\hat x) \int_{\Sph^2}\ud \omega(\hat y) \, 
    \overline{Y_\ell^m(\hat x)} Y_{\ell'}^{m'}(\hat y) \left(\E^{- \I p\cdot(x-y)} - \E^{- \I \sqrt{\mu}\hat p\cdot (x-y)}\right)\,.
\end{aligned}
\end{equation}
Now, using the plane--wave expansion 
$\E^{\I p\cdot x} = 4\pi \sum_{\ell=0}^{\infty} \sum_{m= -\ell}^{\ell} \I^\ell j_\ell(|p||x|) Y_{\ell}^m(\hat p) \overline{Y_\ell^m(\hat x)}$, 
the spherical integrations in $x$ and $y$ may be evaluated as
\[
  16\pi^2 (-\I)^{\ell + \ell'}\left(j_\ell(|p||x|) j_{\ell'}(|p||y|) - j_\ell(\sqrt{\mu}|x|)j_{\ell'}(\sqrt{\mu}|y|)\right) \overline{Y_{\ell}^m(\hat p)} Y_{\ell'}^{m'}(\hat p)
\]
using the orthogonality of the spherical harmonics.
The spherical $p$--integral of this gives a factor $\delta_{\ell \ell'}\delta_{m m'}$ again by orthogonality of the spherical harmonics.
(This shows that $A_{\Delta, \mu}^{(\kappa)}$ is block--diagonal in the angular momentum as claimed.)
We may thus restrict to the case of $\ell = \ell'$ and $m=m'$. Hereinafter, we will write $x$, $y$, and $p$ instead of $|x|$, $|y|$, and $|p|$.

Recall the following bounds on spherical Bessel functions
\[
  \sup_{\ell\in \N_0} \sup_{x\geq 0} |j_\ell(x)| \leq 1\,,
  \qquad
  \sup_{\ell\in \N_0} \sup_{x\geq 0} |j_\ell'(x)| \leq 1\,,
  \qquad
  \sup_{\ell\in \N_0} \sup_{x\geq 0} x^{5/6}|j_\ell(x)| \leq C\,,
\]
where the first one is elementary, the second one follows from \cite[Eq.~10.1.20]{abramowitz}, and the third one may be found in \cite[Eq.~1]{landau} 
(see also Proposition 16 in \cite{henheikTc}).
Adding $\pm j_\ell(px)j_\ell(\sqrt{\mu}y)$ and using these bounds we may estimate
for any $0 < \eps < 5/11$
\begin{equation}\label{eq.bessel.bound}
\begin{aligned}
& \left\vert j_\ell(px) j_{\ell}(py) - j_\ell(\sqrt{\mu}x)j_{\ell}(\sqrt{\mu}y)\right\vert
\\ & \qquad 
  \leq C |p - \sqrt{\mu}|^\eps \left(p^{-\eps} + (\sqrt{\mu})^{-\eps}\right)
  \left(|j_\ell(px)|^{1-11\eps/5} + |j_\ell(\sqrt{\mu}x)|^{1-11\eps/5}\right) 
\\ & \qquad \quad \times
  \left(|j_\ell(py)|^{1-11\eps/5} + |j_\ell(\sqrt{\mu}y)|^{1-11\eps/5}\right).
\end{aligned}
\end{equation}
The radial $p$--integral in \Cref{eq.A.innerprod.full} is then (a constant times) 
\begin{equation}\label{eq.radial.p.int}
  \int_{0}^\infty \ud p \left(\frac{1}{E_{\Delta, \mu}(p)} - \frac{1}{p^2 + \kappa^2\mu}\right) 
    \left(j_\ell(px) j_{\ell}(py) - j_\ell(\sqrt{\mu}x)j_{\ell}(\sqrt{\mu}y)\right)
\end{equation}
Using \Cref{eq.bessel.bound} and changing integration variable $p \rightarrow \sqrt{\mu}p$ we get 
\[
\begin{aligned}
  |\eqref{eq.radial.p.int}|
  & \leq C\mu^{1/2} \int_0^\infty \ud p \, p^2
    \left\vert \frac{1}{\sqrt{(p^2 - 1)^2 + \vert \Delta(\sqrt{\mu} p) /\mu \vert^2 }}- \frac{1}{p^2 + \kappa^2 }\right\vert 
    \vert p-1\vert^{\epsi} \left(\frac{1}{p^\epsi}+ 1\right)
  \\
  & \quad \times 
    \left(|j_\ell(\sqrt{\mu}px)|^{1-11\eps/5} + |j_\ell(\sqrt{\mu}x)|^{1-11\eps/5}\right) 
    \left(|j_\ell(\sqrt{\mu}py)|^{1-11\eps/5} + |j_\ell(\sqrt{\mu}y)|^{1-11\eps/5}\right).
\end{aligned}
\]
Plugging this into \Cref{eq.A.innerprod.full} and using Hölder for the $x$-- and $y$--integrations we thus get
\[
\begin{aligned}
&\big\vert \big\langle f \big\vert A_{\Delta, \mu}^{(\kappa)} \big\vert g \big\rangle \big\vert
\\ & \quad \leq 
C \mu^{1/2} \int_{0}^{\infty} \D p \, p^2 \left\vert \frac{1}{\sqrt{(p^2 - 1)^2 + \vert \Delta(\sqrt{\mu} p) /\mu \vert^2 }}- \frac{1}{p^2 + \kappa^2 }\right\vert 
  \vert p-1\vert^{\epsi} \left(\frac{1}{p^\epsi}+ 1\right) 
  \\
& \qquad \times \int_{\R^3} \D x \,  \vert 
  V (x)\vert\left(\vert j_\ell(\sqrt{\mu}p\vert x \vert) \vert^{2-22\epsi/5} + \vert j_\ell(\sqrt{\mu}\vert x \vert) \vert^{2-22\epsi/5}\right)\,,
\end{aligned}
\]
where we changed back to $x$ denoting a vector in $\R^3$.
By \Cref{lem:3} we may bound the $x$--integral by $\mu^{-\beta^* + \delta}(1 + p^{-\beta^* + \delta})$ for any $\delta > 0$.
Also, $\Vert \Delta \Vert_{L^\infty} = o(\mu)$ by \Cref{prop.delta.lipschitz}.
Hence the $p$--integral will be finite uniformly in $\mu$ for $\mu$ large enough.
We conclude that 
\[
  \norm{A_{\Delta, \mu}^{(\kappa)}}_{\textnormal{op}} \leq C \mu^{-\beta^* + 1/2 + \delta}
\]
for any $\delta > 0$ and for $\mu$ large enough.
Combining this with the bound on $\Vert L_\mu^{(\kappa)}\Vert_{\mathrm{op}}$ from above, we get

\begin{equation} \label{eq:Mbound}
\limsup\limits_{\mu \to \infty} \mu^{\beta^*-1/2-\delta} \left\Vert V^{1/2} M_{\Delta,\mu}^{(\kappa)} \vert V \vert^{1/2} \right\Vert_{\mathrm{op}} =  0
\end{equation}
for any $\delta >0$. 
Also, since
$V^{1/2} \mathfrak{F}_\mu^\dagger\mathfrak{F}_\mu \vert V \vert^{1/2}$ 
is isospectral to $\mathcal{V}_\mu$, so its eigenvalues are given by \Cref{eq:eigenvalues}, one can easily see, using Lemma \ref{lem:3} again, that
\begin{equation} \label{eq:FVFbound}
\limsup\limits_{\mu \to \infty} \mu^{\beta^*-\delta} \left\Vert V^{1/2} \mathfrak{F}_\mu^\dagger
\mathfrak{F}_\mu \vert V \vert^{1/2}\right\Vert_{\mathrm{op}} = 0\,,
\end{equation}
for any $\delta >0$. Finally, by definition of $s^*$ (see \Cref{eq:defsstar}), we get for any $\delta >0$ that
\begin{equation} \label{eq:sVbound}
    \limsup\limits_{\mu \to \infty} \mu^{\min(s^*,2)/2-\delta}\int_{\R^3} \vert V(x)\vert \left( \frac{\sin(\sqrt{\mu} \vert x \vert)}{\sqrt{\mu}\vert x \vert}\right)^2 \D x = 0\,.
\end{equation}
As the last ingredient we need the following Lemma, which provides a bound controlling $\Delta(p)$ in terms of $\Delta (\sqrt{\mu})$. Its proof is given in Section \ref{sec:technical}. 

\begin{lem}
	\label{prop.delta.lip.weak.coupling}
	Suppose $s^* > 1$ and
	let $u(p) = (4\pi)^{-1/2}$ be the constant function on the sphere $\Sph^2$ and let 
	\[
	\hat \varphi(p) = \sqrt{4\pi}\mathfrak{F}V\mathfrak{F}_\mu^\dagger u(p)
	= \frac{1}{(2\pi)^{3/2}} \int_{\Sph^2} \hat V(p - \sqrt{\mu} q) \ud\omega(q)\,,
	\]
	where $\mathfrak{F}$ denotes the usual Fourier transform. 
  Then 
	\[
	\Delta(p) 
	= f(\mu)\left[\hat \varphi(p) + \eta_\mu(p)\right],
	\]
	for some function $f(\mu)$. 
  The function $\eta_\mu$ satisfies
	\begin{equation*}
	\limsup\limits_{\mu \to \infty} \mu^{\beta^* + \min(s^*,2)/4 - 1/2-\delta}\norm{\eta_\mu}_{L^\infty} = 0 \hspace{2.42mm} \text{and} \hspace{2.42mm} \limsup\limits_{\mu \to \infty} \mu^{\beta^* + \min(s^*,2)/2 - 1/2-\delta} \abs{\eta_\mu(\sqrt{\mu})} = 0
	\end{equation*}
	for any $\delta > 0$. 
\end{lem}

\noindent Note that $\hat \varphi(\sqrt\mu) = \sqrt{4\pi}\mathfrak{F}_\mu V\mathfrak{F}_\mu^\dagger u(1) = e_\mu$. 
Now, combining this with Lemmas \ref{lem:3} and \ref{lem:4}, we see that
$\Delta(\sqrt{\mu}) = f(\mu) e_\mu (1 + o(1))$, from which we conclude that
\begin{equation*}
\Delta(p) = \frac{\hat\varphi(p) + \eta_\mu(p)}{e_\mu + \eta_\mu(\sqrt{\mu})} \Delta(\sqrt{\mu})
= \left[1
+ \frac{\hat\varphi(p) - \hat\varphi(\sqrt{\mu})}{e_\mu}
+ \frac{\eta_\mu(p)}{e_\mu}
\right] (1 + o(1)) \Delta(\sqrt{\mu})\,.
\end{equation*}
Now, it is an easy computation to see $|\hat \varphi(p) - \hat\varphi(q)| \leq C \mu^{-1/2} |p-q|$ for all $p,q$. Thus
\begin{equation}
\label{eqn.delta.leq.delta}
|\Delta(p)|  \leq C\left(1 + \mu^{\min(s^*,2)/2-1/2  + \delta}|p-\sqrt{\mu}| 
+ \mu^{\min(s^*,2)/4 - \beta^* +1/2 + \delta}\right) \Delta(\sqrt{\mu})
\end{equation}
for any $\delta >0$, again by means of Lemma \ref{lem:3} and Lemma \ref{lem:4}, assuming that $V$ is admissible. 
So, we get the desired control on $\Delta(p)$ in terms of $\Delta(\sqrt{\mu})$.

The bound on $\eta_{\mu}(\sqrt{\mu})$ is effectively a bound on $\big\langle u \big\vert \mathfrak{F}_\mu^\dagger V M_{\Delta, \mu}^{(\kappa)}V \mathfrak{F}_\mu \big\vert u \big\rangle$.
(This will be clear from the proof.)
For sufficiently large $\mu$ we have
\begin{equation}
\label{eqn.bound.u.m.u}
  \abs{\longip{u}{\mathfrak{F}_\mu^\dagger V M_{\Delta, \mu}^{(\kappa)}V \mathfrak{F}_\mu}{u}} \leq C_\delta \mu^{-\beta^* - \min(s^*, 2)/2 + 1/2 + \delta}
\end{equation}
for any $\delta > 0$. This will be of importance in the perturbation argument in \Cref{prop:5}.

We are now able to prove the second and third equality in Lemma \ref{lem:1}. 
\begin{prop}\label{prop:4} 
Let $V$ be an admissible potential. Then we have
\begin{equation*}
m_\mu^{(\kappa)}(\Delta) = \sqrt{\mu} \left( \log\frac{\mu}{\Delta(\sqrt{\mu})}-2+\kappa \, \frac{\pi}{2} + \log (8) + o(1) \right)
\end{equation*}
in the limit $\mu \to \infty$. 
\end{prop}
\begin{proof}
Computing the angular integral, and substituting $s = \pm \frac{p^2 - \mu}{\mu}$ we get
\[
\begin{aligned}
  m_\mu^{(\kappa)}(\Delta)
  = \frac{\sqrt{\mu}}{2}\Bigg[ 
    & \int_0^1 \left(\frac{\sqrt{1-s} - 1}{\sqrt{s^2 + x_-(s)^2}} + \frac{\sqrt{1 + s} - 1}{\sqrt{s^2 + x_+(s)^2}} 
              - \frac{\sqrt{1-s}}{1-s + \kappa^2} - \frac{\sqrt{1+s}}{1+s+\kappa^2}\right) \ud s
  \\ & + \int_0^1  \left(\frac{1}{\sqrt{s^2 + x_+(s)^2}} + \frac{1}{\sqrt{s^2 + x_-(s)^2}}\right) \ud s
  \\ & + \int_1^\infty \left(\frac{\sqrt{1+s}}{\sqrt{s^2 + x_+(s)^2}} - \frac{\sqrt{1+s}}{1+s+\kappa^2}\right) \ud s
    \Bigg],
\end{aligned}
\]
where $x_{\pm}(s) = \frac{\Delta(\sqrt{\mu}\sqrt{1\pm s})}{\mu}$.
Now, using dominated convergence and $\Vert \Delta \Vert_{L^\infty} = o(\mu)$, it is easy to see that
the first and last integrals converge to
\[
  \int_0^1 \left(\frac{\sqrt{1-s} - 1}{s} + \frac{\sqrt{1 + s} - 1}{s} - \frac{\sqrt{1-s}}{1-s + \kappa^2} - \frac{\sqrt{1+s}}{1+s+\kappa^2}\right) \ud s
\]
and
\[
  \int_1^\infty \left(\frac{\sqrt{1+s}}{s} - \frac{\sqrt{1+s}}{1+s+\kappa^2}\right) \ud s\,,
\]
respectively, in the limit $\mu \to \infty$.
For the middle integral we claim that
\begin{equation} \label{eq:seconintegral}
  \int_0^1 \left(\frac{1}{\sqrt{s^2 + x_\pm(s)^2}} - \frac{1}{\sqrt{s^2 + x_\pm(0)^2}}\right) \ud s \to 0 \quad \text{as} \quad \mu \to \infty\,.
\end{equation}
As in \cite{hs081,lauritsen} this is where we need both the Lipschitz--like bound on $\Delta$ (\Cref{prop.delta.lipschitz}) and the bound controlling $\Delta(p)$ in terms of $\Delta(\sqrt\mu)$ (\Cref{eqn.delta.leq.delta}).
In terms of $x_\pm$, \Cref{prop.delta.lipschitz} reads
\begin{equation} \label{eq:lipschitz}
  |x_\pm(s) - x_\pm(0)| \leq C\mu^{-\delta_r} s\,.
\end{equation}
In terms of $x_\pm$, \Cref{eqn.delta.leq.delta} reads
\begin{equation} \label{eq:comparison}
  x_\pm(s) \leq C(1 + \mu^{\min(s^*,2)/2 +\delta} s + \mu^{\min(s^*,2)/4 - \beta^*+ 1/2 +\delta}) x_\pm(0)\,.
\end{equation}
Now, the integrand in \Cref{eq:seconintegral} is bounded by
\[
  \frac{|x_\pm(s)^2 - x_\pm(0)^2|}{\sqrt{s^2 + x_\pm(s)^2}\sqrt{s^2 + x_\pm(0)^2}\left(\sqrt{s^2 + x_\pm(s)^2} + \sqrt{s^2 + x_\pm(0)^2}\right)}\,.
\]
We introduce a cutoff $\rho \in (0,1)$ and compute the integrals $\int_\rho^1$ and $\int_0^\rho$.
For the first integral we have
\[
\begin{aligned}
  & \int_\rho^1 \frac{|x_\pm(s)^2 - x_\pm(0)^2|}{\sqrt{s^2 + x_\pm(s)^2}\sqrt{s^2 + x_\pm(0)^2}\left(\sqrt{s^2 + x_\pm(s)^2} + \sqrt{s^2 + x_\pm(0)^2}\right)} \ud s
  \\  \leq &\ C \mu^{-\delta_r} \int_\rho^1 \frac{1}{s} \frac{x_\pm(s) + x_\pm(0)}{\sqrt{s^2 + x_\pm(s)^2} + \sqrt{s^2 + x_\pm(0)^2}} \ud s
  \\\leq &\ C \mu^{-\delta_r} |\log \rho|\,.
\end{aligned}
\]
which vanishes for any $\rho \gg \exp\left(-\mu^{\delta_r}\right)$, in particular for $\rho = \mu^{-N}$ for suitable $N > 0$, which we choose here.
For the second integral we have
\[
\begin{aligned}
  & \int_0^\rho \frac{|x_\pm(s)^2 - x_\pm(0)^2|}{\sqrt{s^2 + x_\pm(s)^2}\sqrt{s^2 + x_\pm(0)^2}\left(\sqrt{s^2 + x_\pm(s)^2} + \sqrt{s^2 + x_\pm(0)^2}\right)} \ud s
  \\ \leq &\ C \hspace{-0.5mm}\int_0^\rho  \hspace{-2mm}\mu^{-\delta_r} \hspace{-1mm}\left(1 + \mu^{\min(s^*,2)/4 - \beta^*+1/2+ \delta} + \mu^{\min(s^*,2)/2+\delta} s\right) \hspace{-0.5mm}\frac{x_\pm(0)}{\sqrt{x_\pm(0)^2 + s^2}\left(s + \sqrt{x_\pm(0)^2 + s^2}\right)} \ud s
  \\  \leq &\ C \mu^{\min(s^*,2)/4 - \beta^* - \delta_r +1/2+ \delta} \int_0^\rho \frac{x_\pm(0)}{\sqrt{x_\pm(0)^2 + s^2}\left(s + \sqrt{x_\pm(0)^2 + s^2}\right)} \ud s
  \\ \leq &\  C\mu^{\min(s^*,2)/4 - \beta^* - \delta_r + 1/2 + \delta}\,.
\end{aligned}
\]
Note that for $r=2$, we have $\delta_{r=2} = 3/20$ and thus $  \beta^* - \min(s^*,2)/4-1/2+ 3/20 >0$ for any $s^* > 7/5$ 
(see Remark~\ref{rmk:f(s)}). Also, optimizing this expression in the allowed $r$'s gives the assumption $r > f(s^*)$ given in \Cref{rmk:f(s)}. 
Therefore, also this second integral vanishes as desired by choosing $0< \delta < \beta^* - \min(s^*,2)/4- 7/20$. 
We conclude that 
\[
\begin{aligned}
m_\mu^{(\kappa)}(\Delta) = \frac{\sqrt{\mu}}{2}\Bigg[
  & \int_0^1 \left(\frac{\sqrt{1-s} - 1}{s} + \frac{\sqrt{1 + s} - 1}{s} - \frac{\sqrt{1-s}}{1-s + \kappa^2} - \frac{\sqrt{1+s}}{1+s+\kappa^2}\right) \ud s
  \\ & + \int_0^1 \frac{2}{\sqrt{s^2 + \left(\frac{\Delta(\sqrt{\mu})}{\mu}\right)^2}} \ud s
  + \int_1^\infty \left(\frac{\sqrt{1+s}}{s} - \frac{\sqrt{1+s}}{1+s+\kappa^2}\right) \ud s + o(1)
    \Bigg].
\end{aligned}
\]
This may be computed (perhaps most easily by adding and subtracting the 
corresponding integral with $\kappa = 0$)
as
\[
  m_\mu^{(\kappa)} = \sqrt{\mu} \left(\log \frac{\mu}{\Delta(\sqrt{\mu})} - 2 + \log (8) + \kappa \frac{\pi}{2} + o(1)\right).
  \qedhere
\]
\end{proof}

\noindent
We conclude by showing the third equality of Lemma \ref{lem:1}. 
\begin{prop} \label{prop:5}
Let $V$ be an admissible potential. Then 
\begin{equation*} 
\frac{m_\mu^{(\kappa)}(\Delta)}{\sqrt{\mu}} = - \frac{\pi}{2\sqrt{\mu}b_\mu^{(\kappa)}} + o(1)\,.
\end{equation*}
\end{prop}
\begin{proof} 
Recall that, by the Birman--Schwinger principle the lowest eigenvalue of $B_{\Delta,\mu}$ is $-1$. 
Using the decomposition in \Cref{bsdecomp} and the bound in \Cref{eq:Mbound} we get that
\begin{equation*}
-1 = \lim\limits_{\mu \to \infty} m_\mu^{(\kappa)}(\Delta) \, \mathrm{inf\, spec} \left(V^{1/2} {\mathfrak{F}_{\mu}}^\dagger{\mathfrak{F}_{\mu}} \vert V \vert^{1/2} \right) = \lim\limits_{\mu \to \infty} m_\mu^{(\kappa)}(\Delta) e_\mu \,.
\end{equation*} 
Now, since $s^* >7/5$ we have that $\vert \sqrt{\mu}e_\mu \vert \le C\mu^{-2/5}$ by \Cref{lem:3} (recall~\Cref{eq:eigenvalues} and \Cref{eq:sVbound}).
Thus, by \Cref{prop:4} we conclude that $\Delta(\sqrt{\mu})$ is exponentially small (in some positive power of $\mu$)  as $\mu \to \infty$.

To obtain the next order in the expansion of $m_\mu(\Delta)$, 
we note that $1+V^{1/2} M_{\Delta,\mu}^{(\kappa)} \vert V \vert^{1/2}$ is invertible for $\mu$ large enough by means of \Cref{eq:Mbound}. 
We can thus factorize the Birman--Schwinger operator~\eqref{bsdecomp} as 
	\[
1 + B_{\Delta,\mu} = \left(1 + V^{1/2}M_{\Delta, \mu}^{(\kappa)}|V|^{1/2}\right)
\left( 1 + \frac{m_\mu^{(\kappa)}(\Delta)}{1 + V^{1/2}M_{\Delta,\mu}^{(\kappa)}|V|^{1/2}}
V^{1/2} \mathfrak{F}_\mu^\dagger \mathfrak{F}_\mu |V|^{1/2}\right)\,.
\]
Because $B_{\Delta, \mu}$ has $-1$ as its lowest eigenvalue by the Birman--Schwinger principle, we conclude that, for $\mu$ large enough, the self--adjoint operator
\begin{equation*} 
T_{\Delta, \mu} := m_\mu^{(\kappa)}(\Delta) \mathfrak{F}_\mu|V|^{1/2} 
\frac{1}{1 + V^{1/2}M_{\Delta,\mu}^{(\kappa)}|V|^{1/2}}
V^{1/2}\mathfrak{F}_\mu^\dagger 
\end{equation*}
acting on $L^2(\Sph^2)$ has $-1$ as its lowest eigenvalue since it is isospectral to the right--most operator above.
(This follows from the fact that for operators $A,B$ the operators $AB$ and $BA$ have the same spectrum apart from possibly at $0$.
See also the argument around Equation~(33) in \cite{henheikTc} as well as around Equation~(30) and Equation (47) in \cite{hs081}.)

To highest order $T_{\Delta, \mu}$ is  proportional to $\mathcal{V}_\mu$. 
Since the constant function $u(p) = (4\pi)^{-1/2}$ on $\Sph^2$ is the unique eigenvector of $\mathcal{V}_\mu$ with lowest eigenvalue, 
this is true also for $T_{\Delta, \mu}$ whenever $\mu$ is large enough.

To find the lowest eigenvalue (which is $-1$) we expand the geometric series to first order 
and employ first order perturbation theory. This is completely analogous to the arguments in \cite{hs081} and Equation (34) in \cite{henheikTc}. We obtain
  \begin{equation} \label{eq:impleq}
 \frac{1}{\sqrt{\mu}} m_\mu^{(\kappa)}(\Delta) 
  = \frac{-1}{\mu^{1/2} e_\mu 
    - \mu^{1/2} \big\langle u \big\vert \mathfrak{F}_\mu V M_{\Delta,\mu}^{(\kappa)} V \mathfrak{F}_\mu^\dagger \big\vert u \big\rangle 
    + O(\mu^{-3\beta^* + 3/2+\delta})}
 \end{equation}
for any $\delta >0$ (recall Equations \eqref{eq:Mbound}, \eqref{eq:FVFbound} and \eqref{eqn.bound.u.m.u}). 
The error term in \Cref{eq:impleq} is twofold. 
The first part comes from the expansion of the geometric series. 
The second part comes from first order perturbation theory using the bounds 
\[
  \vert \sqrt{\mu} e_\mu \vert \ge c_\delta \mu^{-\min(s^*,2)/2+1/2-\delta }\quad \text{and} \quad  
\big\vert \mu^{1/2}\big\langle u \big\vert \mathfrak{F}_\mu V M_{\Delta,\mu}^{(\kappa)} V \mathfrak{F}_\mu^\dagger \big\vert u \big\rangle\big\vert 
    \le C_\delta \mu^{-\beta^* - \min(s^*, 2)/2 +1 + \delta} 
\]
 for any $\delta >0$ from \Cref{lem:4} and \Cref{eqn.bound.u.m.u}.
 The error from the series expansion is of order $O(\mu^{-3\beta^* + 3/2+\delta})$ and the error from the perturbation argument is of order $O(\mu^{-2\beta^*-\min(s^*,2)/2 + 3/2+\delta})$
 and is hence dominated by the expansion of the geometric series, since $\beta^* \leq \min(s^*,2)/2$.

 Now, we need to show that 
$\mathfrak{F}_\mu V M_{\Delta,\mu}^{(\kappa)} V \mathfrak{F}_\mu^\dagger$ is close to $\mathcal{W}_\mu^{(\kappa)}$, when evaluated in $\longip{u}{\cdots}{u} $. Therefore, considering their difference, we split the involved radial $p$--integral according to $|p| \le \mu^N$ and $|p| > \mu^N$ for some large $N >0$. The second part is clearly bounded by, e.g., $C \mu^{-N/2}$. For the first part, we have 
 $\Delta(p) \le C \mu^N \Delta(\sqrt{\mu})$ by \Cref{eqn.delta.leq.delta}.
 Using this in combination with the fact that $\Delta(\sqrt{\mu})$ is exponentially small, we find, by dominated convergence and Lipschitz continuity of the involved angular integrals (cf.~Equation (35) in \cite{hs081} and Equation (36) in \cite{henheikTc}), that this part is bounded by $C_D \mu^{-D}$ for any $D > 0$. Since $N>0$ was arbitrary, we conclude that
 \begin{equation} \label{eq:Deltaconv}
\big\vert \big\langle u \big\vert \mathfrak{F}_\mu V M_{\Delta,\mu}^{(\kappa)} V \mathfrak{F}_\mu^\dagger - \mathcal{W}_\mu^{(\kappa)} \big\vert u\big\rangle  \big\vert \le C_D \mu^{-D}
 \end{equation}
 for any $D>0$. Thus, by combining \Cref{eqn.bound.u.m.u} and \Cref{eq:Deltaconv} (recall \Cref{eq:Bmu} and \Cref{eq:bmu}) we get
 \begin{equation} \label{eq:Wmubound}
\big\vert \big\langle u \big\vert \mathcal{W}_\mu^{(\kappa)} \big\vert u \big\rangle \big\vert 
 \le C_\delta \mu^{-\beta^* - \min(s^*,2)/2 +1/2 + \delta}
 \end{equation}
 for any $\delta > 0$. (In particular $b_\mu^{(\kappa)} < 0$ for large $\mu$. This was also shown in \cite{henheikTc}.)

 In particular, combining Equations \eqref{eq:impleq}, \eqref{eq:Deltaconv} and \eqref{eq:Wmubound}, 
 we get again by a perturbation theory argument that
 \begin{equation*}
 \frac{1}{\sqrt{\mu}} m_\mu^{(\kappa)}(\Delta) = - \frac{\pi}{2 \sqrt{\mu} b_\mu^{(\kappa)}} + O(\mu^{-3\beta^* +  \min(s^*,2)+1/2+\delta})\,,
 \end{equation*}
 for any $\delta > 0$. Since $3\beta^* -  \min(s^*,2) -1/2 > 0$ 
 we conclude the desired. 
\end{proof}

 \subsection{Proofs of Auxiliary Lemmas} \label{sec:technical}
In this Subsection, we prove the auxiliary Lemmas \ref{lem:minimizer1}, \ref{prop.delta.lipschitz}, and \ref{prop.delta.lip.weak.coupling}. 
\begin{proof}[Proof of Lemma \ref{lem:minimizer1}]
	First we show 
	\begin{equation}\label{bound.alpha.first}
	\norm{\alpha}_{H^1}^2 \leq C \norm{\alpha}_{L^2}^2 + C \mu^{3/2}.
	\end{equation}
	Since $V\in L^{3/2}({\R^3})$ we have by Sobolev's inequality \cite[Thm.~8.3]{liebloss} $\inf \spec \left(\frac{p^2}{2} + V\right) > -\infty$.
	Thus, using $\sqrt{1-4x^2}\leq 1 - 2x^2$ and $\hat \alpha \leq 1/2$ we get
	\begin{align*}
	\mathcal{F}(\alpha)
	& = 
	\frac{1}{2} \int_{\R^3} |p^2 - \mu| \left(1 - \sqrt{1 - 4\hat \alpha(p)^2}\right) \ud p + \int_{\R^3}V(x) |\alpha(x)|^2 \ud x
	\\
	& \geq 
	\int_{\R^3} (p^2 - \mu) \hat \alpha(p)^2 \ud p + \int_{\R^3} V(x)|\alpha(x)|^2 \ud x
	\\
	& =
	\longip{\alpha}{\frac{p^2}{2} + V}{\alpha} + \int_{\R^3} \left(\frac{p^2}{2} - \mu\right) \hat\alpha(p)^2 \ud p
	\\
	&\geq 
	\frac{1}{4}\norm{\alpha}_{H^1}^2 - C \norm{\alpha}_{L^2}^2 + \int_{\R^3} \left(\frac{p^2}{4} - \mu - \frac{1}{4}\right) \hat \alpha(p)^2 \ud p
	\\
	& \geq 
	\frac{1}{4}\norm{\alpha}_{H^1}^2 - C\norm{\alpha}_{L^2}^2 
	- \frac{1}{4} \int_{\R^3} \left[\frac{p^2}{4} - \mu - \frac{1}{4}\right]_{-} \ud p
	\\
	& \geq 
	\frac{1}{4}\norm{\alpha}_{H^1}^2 - C\norm{\alpha}_{L^2}^2 - C \mu^{3/2}\,,
	\end{align*}
	which gives the desired. Now we show that 
	\[
	\norm{\hat \alpha \mathbbm{1}_{\{|p| < t\}}}_{L^2} \leq C\norm{\hat \alpha \mathbbm{1}_{\{|p| > t\}}}_{L^2} + C\mu^{2\delta - 1} \norm{\alpha}_{H^1}^2,
	\]
	for $t = \mu^{\delta}$ and $0 < \delta < 1/2$.
	
	To see this, we split the integrals in the functional $\mathcal{F}$ according to small or large momentum $p$ and compute
	\begin{align*}
	\mathcal{F}(\alpha) &= \frac{1}{2} \int_{\R^3} |p^2 - \mu| \left(1 - \sqrt{1 - 4\hat \alpha(p)^2}\right) \ud p
	+ \int_{\R^3} V(x) |\alpha(x)|^2 \ud x
	\\
	& \geq 
	\int_{|p| < t} |p^2 - \mu| \hat \alpha(p)^2 \ud p
	+ \int_{|p| > t} |p^2 - \mu| \hat \alpha(p)^2 \ud p
	\nopagebreak[4]\\ & \hspace{6cm}+ \frac{1}{(2\pi)^{3/2}} \iint_{\R^3\times \R^3}  \hat \alpha(p) \hat V(p - q) \hat \alpha(q) \ud p \ud q
	\\
	&  \geq
	\mu \norm{\hat \alpha \mathbbm{1}_{\{|p| < t\}}}_{L^2}^2 - \norm{p^2\hat \alpha \mathbbm{1}_{\{|p| < t\}}}_{L^2}^2
	+ \longip{\hat \alpha \mathbbm{1}_{\{|p| > t\}}}{p^2 + V}{\hat \alpha \mathbbm{1}_{\{|p| > t\}}}
	- \mu \norm{\hat \alpha \mathbbm{1}_{\{|p| > t\}}}_{L^2}^2 \\
	& \quad + \frac{1}{(2\pi)^{3/2}} \left[
	\iint_{|p|, |q| < t} \hspace{-3mm} \hat \alpha(p) \hat V(p - q) \hat \alpha(q) \ud p \ud q  + 2 \iint_{|p| < t, |q|>t} \hspace{-3mm}\hat \alpha(p) \hat V(p - q) \hat \alpha(q) \ud p \ud q
	\right].
	\end{align*}
	Note that, again by Sobolev's inequality \cite[Thm.~8.3]{liebloss}, we have
	\[
	\longip{\hat \alpha \mathbbm{1}_{\{|p| > t\}}}{p^2 + V}{\hat \alpha \mathbbm{1}_{\{|p| > t\}}} \geq 
	- C\norm{\hat \alpha \mathbbm{1}_{\{|p| > t\}}}_{L^2}^2\,.
	\]
	Moreover, by application of Young's inequality \cite[Thm.~4.2]{liebloss} we obtain
	\[
	\begin{aligned}
	\int_{|p| < t} \int_{|q| < t} \hat \alpha(p) \hat V(p - q) \hat \alpha(q) \ud p \ud q
	& \geq - \norm{\hat V}_{L^3}\norm{\hat \alpha \mathbbm{1}_{\{|p| < t\}}}_{L^{6/5}}^2
	\\ & 
	\geq - C \left(t^3\right)^{2\cdot(5/6 - 1/2)} \norm{\hat \alpha \mathbbm{1}_{\{|p| < t\}}}_{L^2}^2
	\\ &
	= - C \mu^{2\delta} \norm{\hat \alpha \mathbbm{1}_{\{|p| < t\}}}_{L^2}^2
	\end{aligned}
	\]  
	and
	\[
	\begin{aligned}
	\int_{|p| < t} \int_{|q| > t} \hat \alpha(p) \hat V(p - q) \hat \alpha(q) \ud p \ud q
	& 
	\geq - \norm{\hat \alpha \mathbbm{1}_{\{|p| < t\}}}_{L^1} \norm{\hat V}_{L^3} \norm{\hat \alpha \mathbbm{1}_{\{|p| > t\}}}_{L^{3/2}}
	\\ & 
	\geq - C t^{3/2} \norm{\alpha}_{H^1} \norm{\hat \alpha \mathbbm{1}_{\{|p| < t\}}}_{L^2}
	\\ & 
	= - C \mu^{3\delta/2} \norm{\alpha}_{H^1} \norm{\hat \alpha \mathbbm{1}_{\{|p| < t\}}}_{L^2}\,,
	\end{aligned}
	\]
	where we used that $\norm{\hat g}_{L^{3/2}} \leq C \norm{g}_{H^1}$.
	Thus we arrive at
	\[
	\mathcal{F}(\alpha) \geq c \mu \norm{\hat \alpha \mathbbm{1}_{\{|p| < t\}}}_{L^2}^2
	- C_1 \mu^{3\delta/2} \norm{\alpha}_{H^1} \norm{\hat \alpha \mathbbm{1}_{\{|p| < t\}}}_{L^2}
	-   C_2 \mu \norm{\hat \alpha \mathbbm{1}_{\{|p| > t\}}}_{L^2}^2,
	\]
	where we absorbed all non-leading terms in these. 
	This is a second degree polynomial in $\norm{\hat \alpha \mathbbm{1}_{\{|p| < t\}}}_{L^2}$
	and thus the value of $\norm{\hat \alpha \mathbbm{1}_{\{|p| < t\}}}_{L^2}$ lies between the roots, i.e.
	\[
	\begin{aligned}
	\norm{\hat \alpha \mathbbm{1}_{\{|p| < t\}}}_{L^2}
	& \leq \frac{
		C_1 \mu^{3\delta/2} \norm{\alpha}_{H^1} 
		+ \sqrt{C_1^2\mu^{3\delta} \norm{\alpha}_{H^1}^2 + 4 c \, C_2 \mu^2 \norm{\hat \alpha \mathbbm{1}_{\{|p| > t\}}}_{L^2}^2}
	}{2c\mu}
	\\
	& \leq C \norm{\hat \alpha \mathbbm{1}_{\{|p| > t\}}}_{L^2} + C \mu^{3\delta/2 - 1} \norm{\alpha}_{H^1}.
	\end{aligned}
	\]
	From the estimate
	\[
	\norm{\hat \alpha \mathbbm{1}_{\{|p| > t\}}}_{L^2}^2
	= \int_{|p| > t} \hspace{-2mm}\hat \alpha(p)^2 \ud p
	\leq \int_{|p| > t} \hspace{-2mm}\hat \alpha(p)^2 \frac{1 + p^2}{1 + t^2}\ud p
	\leq \frac{1}{1 + t^2}\norm{\alpha}_{H^1}^2
	\leq C \mu^{-2\delta}\norm{\alpha}_{H^1}^2\,,
	\]
	we conclude that
	\[
	\norm{\alpha}_{L^2}^2
	= \norm{\hat \alpha \mathbbm{1}_{\{|p| < t\}}}_{L^2}^2 + \norm{\hat \alpha \mathbbm{1}_{\{|p| > t\}}}_{L^2}^2
	\leq C \left(\mu^{-2\delta} + \mu^{3\delta - 2}\right) \norm{\alpha}_{H^1}^2.
	\]
	Choosing the optimal $\delta = 2/5$ we get $\norm{\alpha}_{L^2} \leq C \mu^{-2/5} \norm{\alpha}_{H^1}$, which, in combination with \Cref{bound.alpha.first}, yields
	\[
	\norm{\alpha}_{H^1}^2 \leq C \mu^{-4/5} \norm{\alpha}_{H^1}^2 + \mu^{3/2}.
	\]
	Hence $\norm{\alpha}_{H^1} \leq C\mu^{3/4}$
	and thus also $\norm{\alpha}_{L^2} \leq C \mu^{-2/5} \norm{\alpha}_{H^1} \leq C \mu^{7/20}$
	for sufficiently large~$\mu$.
\end{proof}

\noindent We now turn to the proof of Lemma \ref{prop.delta.lipschitz}.

\begin{proof}[Proof of Lemma \ref{prop.delta.lipschitz}]
	Let $t = \frac{5}{2} - \frac{3}{r}$. Then we have
	\[
	\norm{\Delta}_{L^\infty} \leq C \norm{V\alpha}_{L^1} 
	\leq C \norm{V}_{L^{r}} \norm{\alpha}_{L^{{r}'}}
	\leq C \norm{\alpha}_{L^2}^t \norm{\alpha}_{L^6}^{1-t}
	\leq C \mu^{\frac{15-8t}{20}}
	= C \mu^{\frac{24-5r}{20r}}
	\]
	by Sobolev's inequality \cite[Thm.~8.3]{liebloss}.
	For the difference note that $\Delta(p) - \Delta(q)$ is (proportional to) the Fourier transform of 
	$V(x)\left(1 - \E^{\I (p-q)\cdot x}\right) \alpha(x)$.
	Then
	\[
	\norm{V(x)\left(1 - \E^{\I (p-q)x}\right)}_{L^{r}}^{r}
	= \int_{\R^3} |V(x)|^{r} \abs{1 - \E^{\I (p-q)\cdot x}}^{r} \ud x
	\leq C \int_{\R^3} |V(x)|^{r} |p-q|^{r} |x|^{r} \ud x.
	\]
	Using radiality of $\Delta$, the same argument as before gives the desired.
\end{proof}

\noindent Finally, we give the proof of Lemma \ref{prop.delta.lip.weak.coupling}. 

\begin{proof}[Proof of Lemma \ref{prop.delta.lip.weak.coupling}] 
 Recall from the factorization of the Birman--Schwinger operator in the proof of Proposition \ref{prop:5}, that the self--adjoint operator 
	\begin{equation*} 
	m_\mu^{(\kappa)}(\Delta) \mathfrak{F}_\mu|V|^{1/2} 
	\frac{1}{1 + V^{1/2}M_{\Delta,\mu}^{(\kappa)}|V|^{1/2}}
	V^{1/2}\mathfrak{F}_\mu^\dagger 
	\end{equation*}
	acting on $L^2(\Sph^2)$ has $-1$ as its lowest eigenvalue and $u(p) = (4\pi)^{-1/2}$ is the unique eigenvector 
  with lowest eigenvalue for $\mu$ large enough.
  Hence, one can easily see that
	\[
	\frac{1}{1 + V^{1/2}M_{\Delta,\mu}^{(\kappa)}|V|^{1/2}} V^{1/2} \mathfrak{F}_\mu^\dagger u
	\]
	is an eigenvector of $B_{\Delta,\mu}$ for the lowest eigenvalue and thus proportional to $V^{1/2}\alpha$. 
  By expanding 
  $\frac{1}{1+x} = 1 - \frac{x}{1+x}$ we conclude that $\Delta = f(\mu)[\hat \varphi + \eta_\mu]$, where
	\[
	\eta_\mu 
	= - \sqrt{4\pi}\mathfrak{F}|V|^{1/2} 
  \frac{V^{1/2} M_{\Delta,\mu}^{(\kappa)} |V|^{1/2}}{1 + V^{1/2} M_{\Delta,\mu}^{(\kappa)} |V|^{1/2}}
  V^{1/2} \mathfrak{F}_\mu^\dagger u\,,
	\]
	which can easily be bounded as 
	\begin{equation*}
	\norm{\eta_\mu}_{L^\infty} 
  \leq C \norm{V}_{L^1}^{1/2} \norm{V^{1/2} M_{\Delta,\mu}^{(\kappa)} |V|^{1/2}}_{\mathrm{op}} \norm{V^{1/2} \mathfrak{F}_\mu^\dagger u}_{L^2}\,.
	\end{equation*}
	For $\vert p \vert = \sqrt{\mu}$, we first note that $\hat\varphi(\sqrt{\mu}) = \sqrt{4\pi}\mathfrak{F}_\mu V \mathfrak{F}_\mu^\dagger u(1) = e_\mu$.
	Similarly, since $\eta_\mu$ is radial, we have that
	\[
  \begin{aligned}
	\eta_\mu(\sqrt{\mu}) 
	& = \frac{1}{4\pi} \int_{\Sph^2} \eta_\mu(\sqrt{\mu}q) \ud\omega(q)
    = - \longip{u}{\mathfrak{F}_\mu |V|^{1/2} 
    \frac{V^{1/2} M_{\Delta,\mu}^{(\kappa)} |V|^{1/2}}{1 + V^{1/2} M_{\Delta,\mu}^{(\kappa)} |V|^{1/2}}
    V^{1/2} \mathfrak{F}_\mu^\dagger}{u}
  \end{aligned}
	\]
and we can thus bound
	\[
	\abs{\eta_\mu(\sqrt{\mu})}
	\leq C \norm{V^{1/2} M_{\Delta,\mu}^{(\kappa)} |V|^{1/2}}_{\mathrm{op}} \norm{|V|^{1/2} \mathfrak{F}_\mu^\dagger u}_{L^2}^2.
	\]
It remains to check that 
	\[
	\norm{|V|^{1/2} \mathfrak{F}_\mu^\dagger u}_{L^2}^2  = C \int_{\R^3} |V(x)|  \left\vert \int_{\Sph^2} e^{i\sqrt{\mu}p\cdot x} \frac{1}{\sqrt{4\pi}}\ud \omega(p) \right\vert^2 \D x= C \int_{\R^3} |V(x)|\left(\frac{\sin\sqrt{\mu}|x|}{\sqrt{\mu}|x|}\right)^2 \D x\,.
	\]
Now the claim follows by application of Equations \eqref{eq:Mbound} and \eqref{eq:sVbound}.
\end{proof}

\bigskip 

\noindent {\it Acknowledgments.} We are grateful to Robert Seiringer for helpful discussions and many valuable comments on an earlier version of the manuscript. 
J.H. acknowledges partial financial support by the ERC Advanced Grant ``RMTBeyond” No.~101020331.
\bigskip 

\noindent {\it Data Availability.} Data sharing is not applicable to this article as no datasets were generated or analyzed during the current study.

\bigskip 

\noindent {\it Conflicts of Interest.} The authors have no financial or non--financial interests to disclose.


\begin{thebibliography}{00} 
   \bibitem{abramowitz} M.~Abramowitz, I.~A.~Stegun (eds.). \textit{Handbook of mathematical functions with formulas, graphs, and mathematical tables}.
    US Government printing office, 10th edition, 1972.

   	\bibitem{bcs} J.~Bardeen, L.~N.~Cooper, J.~Schrieffer. Theory of Superconductivity.  \textit{Phys. Rev.} 108, 1175--1204, 1957. 
   	
	\bibitem{bloch} I.~Bloch, J.~Dalibard, W.~Zwerger. Many-Body Physics with Ultracold Gases. \textit{Rev. Mod. Phys.} 80, 885--964, 2008.
	
	\bibitem{chen} Q.~Chen, J.~Stajic, S.~Tan, K.~Levin. BCS–BEC Crossover: From High Temperature Superconductors to Ultracold Superfluids. \textit{Phys. Rep.} 412, 1–88, 2005. 
	
   	\bibitem{cueninmerz} J.--C.~Cuenin, K.~Merz. Weak Coupling Limit for Schrödinger Operators with Degenerate Kinetic Energy for a Large Class of Potentials. \textit{Lett. Math. Phys.} 111, 46, 2021.
   	
   	\bibitem{fhns} R.~L.~Frank, C.~Hainzl, S.~Naboko, R.~Seiringer. The Critical Temperature for the BCS Equation at Weak Coupling. \textit{The Journal of Geometric Analysis} 17, 559--567, 2007.
   	
\bibitem{gorkov}  L.~P.~Gor’kov, T.~K.~Melik--Barkhudarov. Contributions to the Theory of Superfluidity in an Imperfect Fermi Gas. \textit{Soviet Physics JETP} 13, 1018, 1961. 

\bibitem{hhss} C.~Hainzl, E.~Hamza, R.~Seiringer, J.~P.~Solovej. The BCS Functional for General Pair Interaction. \textit{Comm. Math. Phys.} 281, 349--367, 2008. 

\bibitem{hs081} C.~Hainzl, R.~Seiringer. Critical Temperature and Energy Gap for the BCS Equation. \textit{Phys. Rev. B}, 77, 184517, 2008.

\bibitem{hs08} C.~Hainzl, R.~Seiringer. The BCS Critical Temperature for Potentials with Negative Scattering Length. \textit{Lett. Math. Phys.} 84, 99--107, 2008. 

\bibitem{hs10} C.~Hainzl, R.~Seiringer. Asymptotic Behavior of Eigenvalues of Schrödinger Type Operators with Degenerate Kinetic Energy. \textit{Math. Nachr.} 283, 489--499, 2010. 

\bibitem{hs15} C.~Hainzl, R.~Seiringer. The Bardeen-Cooper--Schrieffer Functional of Superconductivity and its Mathematical Properties. \textit{J. Math. Phys.} 57, 021101, 2016. 
	
\bibitem{hall} B.~C.~Hall. Quantum Theory for Mathematicians. \textit{Springer}, 2013.

	\bibitem{henheikTc} J.~Henheik. The BCS Critical Temperature at High Density.  \textit{Math.~Phys.~Anal.~Geom.} 25, 2022. 
	
\bibitem{landau} L.~J.~Landau. Bessel Functions: Monotonicity and Bounds. \textit{Journal of the London Mathematical Society} 61, 197--215, 2000.

\bibitem{langmann} E.~Langmann, C.~Triola, A.~V.~Balatasky. Ubiquity of Superconducting Domes in the Bardeen--Cooper--Schrieffer Theory with Finite--Range Potentials. \textit{Phys. Rev. Lett.} 122, 157001, 2019. 
	
\bibitem{lauritsen} A.~B.~Lauritsen. The BCS Energy Gap at Low Density. \textit{Lett. Math. Phys.} 111, 20, 2021. 

\bibitem{liebloss} E.~H.~Lieb, M.~Loss. Analysis. \textit{Amer. Math. Soc.}, 2nd edition, 2001. 
\end{thebibliography}
\end{document}